\documentclass{patmorin}
\usepackage{graphicx}
\usepackage{textcase}
\usepackage{amsmath,amsthm}
\usepackage{pat}

\newcommand{\mst}{\mathit{MST}}

\title{\MakeTextUppercase{A Note on Interference in Random Networks}%
   \thanks{This work was partly funded by NSERC.}}
\author{Luc Devroye\thanks{School of Computer Science, McGill University}
        \ and
        Pat Morin\thanks{School of Computer Science, Carleton University}}

\begin{document}
\maketitle

\begin{abstract}
  The (maximum receiver-centric) interference of a geometric graph
  (von Rickenbach \etal\ (2005)) is studied.  It is shown that,
  with high probability, the following results hold for a set, $V$,
  of $n$ points independently and uniformly distributed in the unit
  $d$-cube, for constant dimension $d$: (1)~there exists a connected
  graph with vertex set $V$ that has interference $O((\log n)^{1/3})$;
  (2)~no connected graph with vertex set $V$ has interference $o((\log
  n)^{1/4})$; and (3)~the minimum spanning tree of $V$ has interference
  $\Theta((\log n)^{1/2})$.
\end{abstract}

\section{Introduction}

Von Rickenbach \etal\ \cite{vR05,rwz09} introduce the notion of (maximum
receiver-centric) interference in wireless networks and argue that
topology-control algorithms for wireless networks should explicitly take
this parameter into account.  Indeed, they show that the minimum spanning
tree, which seems a natural choice to reduce interference, can be very
bad; there exists a set of node locations in which the minimum spanning
tree of the nodes produces a network with maximum interference that is
linear in the number, $n$, of nodes, but a more carefully chosen network
has constant maximum interference, independent of $n$.  These results are,
however, \emph{worst-case}; the set of node locations that achieve this
are very carefully chosen.  In particular, the ratio of the distance
between the furthest and closest pair of nodes is exponential in the
number of nodes.

The current paper continues the study of maximum interference, but in a
model that is closer to a typical case.  In particular, we consider what
happens when the nodes are distributed uniformly, and independently,
in the unit square.  This distribution assumption can be used to
approximately model the unorganized nature of ad-hoc networks and is
commonly used in simulations of such networks \cite{tma09}. Additionally, some
types of sensor networks, especially with military applications, are
specifically designed to be deployed by randomly placing (scattering)
them in the deployment area. This distribution assumption models these
applications very well.

Our results show that the maximum interference, in this case, is very
far from the worst-case.  In particular, for points independently and
uniformly distributed in the unit square, the maximum interference of the
minimum spanning tree grows only like the square root of the logarithm
of the number of nodes.  That is, the maximum interference is \emph{not
even logarithmic} in the number of nodes.  Furthermore, a more carefully
chosen network topology can reduce the maximum interference further still,
to the cubed root of the logarithm of $n$.

\subsection{The Model}

Let $V=\{x_1,\ldots,x_n\}$ be a set of $n$ points in $\R^d$ and let
$G=(V,E)$ be a simple undirected graph with vertex set $V$.  The graph
$G$ defines a set, $B(G)$, of closed balls $B_1,\ldots,B_n$, where $B_i$
has center $x_i$ and radius
\[
   r_i = \max\{\|x_ix_j\| : x_ix_j\in E\} \enspace .
\]
(Here, and throughout, $\|xy\|$ denotes the Euclidean distance between
points $x$ and $y$.)  In words, $B_i$ is just large enough to enclose
all of $x_i$'s neighbours in $G$.  The \emph{(maximum receiver-centric)
interference} at a point, $x$, is the number of these balls that contain
$x$, i.e.,
\[
    I(x,G) = |\{B\in B(G) : x\in B\}| \enspace .
\]
The \emph{(maximum receiver-centric) interference} of $G$ is the maximum
interference at any vertex of $G$, i.e.,
\[
   I(G) = \max\{I(x,G) : x\in V\} \enspace .
\]
\figref{interference} shows an example of a geometric graph $G$ and the
balls $B(G)$.  Each node, $x$, is labelled with $I(x,G)$.

\begin{figure}
  \begin{center}
    \includegraphics{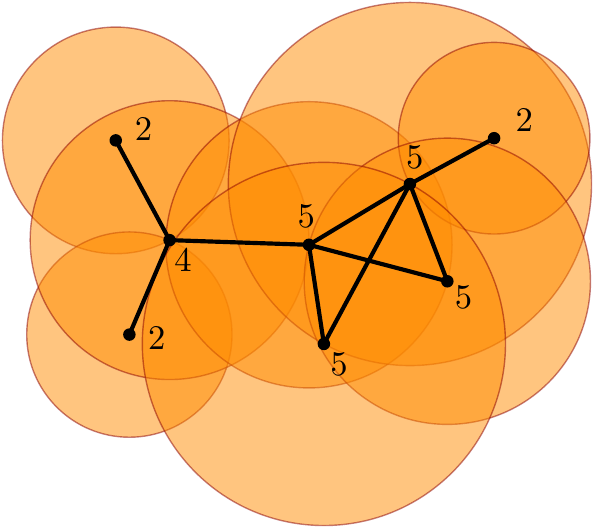}
  \end{center}
  \caption{A geometric graph $G$ with $I(G)=5$.}
  \figlabel{interference}
\end{figure}

One of the goals of network design is to build, given $V$, a connected
graph $G=(V,E)$ such that $I(G)$ is minimized.  Thus, it is natural to consider interference as a property of the given point set, $V$, defined as
\[
  I(V) = \min\{I(G) : \mbox{$G=(V,E)$ is connected}\} \enspace .
\]
A \emph{minimum spanning tree} of $V$ is a connected graph, $\mst(V)$,
of minimum total edge length.  Minimum spanning trees are a natural
choice for low-interference graphs.  The purpose of the current paper
is to prove the following results (here, and throughout, the phrase
\emph{with high probability} means with probability that approaches 1
as $n\rightarrow\infty$):
\begin{thm}\thmlabel{main}
  Let $V$ be a set of $n$ points independently and uniformly distributed
  in $[0,1]^d$.  With high probability,
\begin{enumerate}
  \item  $I(\mst(V))\in O((\log n)^{1/2})$;
  \item  $I(V)\in O((\log n)^{1/3})$, for $d\in\{1,2\}$; and 
  \item $I(V)\in O((\log n)^{1/3}(\log\log n)^{1/2})$, for $d\ge 3$.
\end{enumerate}
\end{thm}
\begin{thm}\thmlabel{lower-bound}
  Let $V$ be a set of $n$ points independently and uniformly distributed
  in $[0,1]^d$.  With high probability, 
\begin{enumerate}
  \item $I(\mst(V))\in\Omega((\log n)^{1/2})$ 
  \item $I(V)\in \Omega((\log n)^{1/4})$.
\end{enumerate}
\end{thm}
\setcounter{thm}{3}

\subsection{Related Work}

This section surveys previous work on the problem of bounding the
interference of worst-case and random point sets.  A summary of the
results described in this section is given in \figref{related}.  In the
statements of all results in this section, $|V|=n$.

\begin{figure}
\begin{center}
  \begin{tabular}{|l|l|r@{ }l|}\hline
    Ref. & Dimension & \multicolumn{2}{c|}{Statement} \\ \hline
    \cite{vR05} & $d\ge 1$ & there exists $V$ s.t.\ & $I(V)\in \Omega(n^{1/2})$ \\
    \cite{vR05} & $d=1$ & for all $V$, & $I(V)\in O(n^{1/2})$ \\
    \cite{ht08} & $d=2$ & for all $V$, & $I(V)\in O(n^{1/2})$ \\
    \cite{ht08} & $d\ge 3$ & for all $V$, & $I(V)\in O((n\log n)^{1/2})$ \\
    \cite{kkmns10} & $d= 1$ & for $V$ i.u.d. in $[0,1]$, & $I(\mst(V))\in \Theta((\log n)^{1/2})$ w.h.p. \\
    \cite{kdh11} & $d\ge 2$ & for $V$ i.u.d. in $[0,1]^d$, & $I(\mst(V))\in O(\log n)$ w.h.p.  \\
    Here & $d\ge 1$ & for $V$ i.u.d. in $[0,1]^d$, & $I(\mst(V))\in \Theta((\log n)^{1/2})$ w.h.p.  \\ 
    \cite{kkmns10,vR05} & $d = 1$ & for $V$ i.u.d. in $[0,1]$, & $I(V)\in\Omega((\log n)^{1/4})$ w.h.p.  \\
    Here & $d\ge 1$ & for $V$ i.u.d. in $[0,1]^d$,  & $I(V)\in \Omega((\log n)^{1/4})$ w.h.p.  \\ 
    Here & $d\in\{1,2\}$ & for $V$ i.u.d. in $[0,1]^d$, & $I(V)\in O((\log n)^{1/3})$ w.h.p.  \\
    Here & $d\ge 3$ & for $V$ i.u.d. in $[0,1]^d$, & $I(V)\in O((\log n)^{1/3}(\log\log n)^{1/2})$ w.h.p.  \\
  \hline
  \end{tabular}
\end{center}
\caption{Previous and new results on interference in geometric networks.}
\figlabel{related}
\end{figure}

The definition of interference used in this paper was introduced by
von~Rickenbach \etal\ \cite{vR05} who proved upper and
lower bounds on the interference of one dimensional point sets:
\begin{thm}[von Rickenbach \etal\ 2005]\thmlabel{sqrtnlower}
For any $d\ge 1$, there exists $V\subset\R^{d}$ such
that $I(V)\in\Omega(n^{1/2})$.
\end{thm}
The point set, $V$, in this lower-bound consists of any sequence of
points $x_1,\ldots,x_n$, all on a line, such that $\|x_{i+1}x_i\| \le (1/2)\|x_{i}x_{i-1}\|$,
for all $i\in\{2,\ldots,n-1\}$.  That is, the gaps between consecutive
points decrease exponentially.

This lower bound is matched by an upper-bound:
\begin{thm}[von Rickenbach \etal\ 2005]\thmlabel{twod-upper}
For all $V\subset\R$, $I(V)\in O(n^{1/2})$.
\end{thm}
The upper bound in \thmref{twod-upper} is obtained by selecting $n^{1/2}$
vertices to act as \emph{hubs}, connecting the hubs into any connected
network and then having each of the remaining nodes connect to its
nearest hub.  This idea was extended to two and higher dimensions
by Halld\'orsson and Tokuyama \cite{ht08}, by using a special type of
$(n^{-1/2})$-net as the set of hubs:
\begin{thm}[Halld\'orsson and Tokuyama 2008]\thmlabel{sqrtn2d}
For all $V\subset\R^d$,
\begin{enumerate}
\item $I(V)\in O(n^{1/2})$ for $d=2$; and
\item $I(V)\in O((n\log n)^{1/2})$, for $d\ge 3$.
\end{enumerate}
\end{thm}

Several authors have shown that the interference of a point set is
related to the (logarithm of) the ratio between the longest and shortest
distance defined by the point set.  In particular, different versions
of the following theorem have been proven by Halld\'orsson and Tokuyama
\cite{ht08}; Khabbazian, Durocher, and Haghnegahdar
\cite{kdh11}; and Maheshwari, Smid, and Zeh \cite{msz11}:
\begin{thm}[Halld\'orsson and Tokuyama
2008; Khabbazian, Durocher, and Haghnegahdar
2011; Maheshwari, Smid, and Zeh 2011]\thmlabel{log}
  For any constant $d\ge 1$ and for all $V\subset\R^d$,
  $I(V)=O(\log D)$, where $D=\max\{\|xy\|: \{x,y\}\subseteq V\}/\min\{\|xy\|:
  \{x,y\}\subseteq V\}$.
\end{thm}
At least two of the proofs of \thmref{log} proceed by showing that
$I(\mst(V))=O(\log D)$.  A strengthening of this theorem  is that the
numerator in the definition of $D$ can be replaced with the length of
the longest edge in $\mst(V)$ \cite{kdh11,msz11}.

\thmref{log} suggests that point sets with very high interference are
unlikely to occur in practice.  This intuition is born out by the results
of Kranakis \etal\ \cite{kkmns10}, who show that high interference is
unlikely to occur in random point sets in one dimension:
\begin{thm}[Kranakis \etal\ 2010]\thmlabel{sqrtlogn}
  Let $V$ be a set of $n$ points independently and uniformly distributed
  in $[0,1]$.  Then, with high probability, $I(\mst(V))\in \Theta((\log
  n)^{1/2})$.
\end{thm}
Note that, in this one-dimensional case, the minimum spanning tree,
$\mst(V)$, is simply a path that connects the points of $V$ in order,
from left to right.  Taken together, Part 1 of Theorems~\ref{thm:main}
and \ref{thm:lower-bound} generalize \thmref{sqrtlogn} to arbitrary
constant dimensions $d\ge 1$.

In higher dimensions, Khabbazian, Durocher, and Haghnegahdar \cite{kdh11}
use their version of \thmref{log} to show that minimum spanning trees
of random point sets have at most logarithmic interference.
\begin{thm}[Khabbazian, Durocher, and Haghnegahdar 2011]\thmlabel{durocher}
  Let $V$ be a set of $n$ points independently and uniformly distributed
  in $[0,1]^d$.  Then, with high probability, $I(\mst(V))\in O(\log n)$.
\end{thm}
Part 1 of \thmref{main} improves the upper bound in \thmref{durocher} to
$O((\log n)^{1/2})$ and Part 1 of \thmref{lower-bound} gives a matching
lower bound.

The second parts of Theorems~\ref{thm:main} and \ref{thm:lower-bound}
show that minimum spanning trees do not minimize interference, even for
random point sets.  For random point sets, one can construct networks with
interference $O((\log n)^{1/3})$ and the best networks have interference
in $\Omega((\log n)^{1/4})$.

The remainder of this paper is devoted to proving Theorems~\ref{thm:main}
and \ref{thm:lower-bound}.  For ease of exposition, we only present these
proofs for the case $d=2$ though they generalize, in a straightforward
way, to arbitrary (constant) dimensions.

\section{Proof of the Upper Bounds (\thmref{main})}

In this section, we prove \thmref{main}.  However, before we do this,
we state a slightly modified version of \thmref{log} that is needed in
our proof.

\begin{lem}\lemlabel{log}
  Let $V\subset\R^d$, let $r>0$, and let $\mst^r(V)$ denote the subgraph
  of $\mst(V)$ containing only the edges whose length is in $(r,2r]$.
  Then $I(\mst^r(V))\in O(1)$.
\end{lem}

\begin{proof}
(This proof is similar to the proof of Lemma~3 in Ref.~\cite{msz11}.)
Let $x$ be any point in $\R^d$ and let $B$ the set of all
balls in $B(\mst^r(V))$ that contain $x$ so that, by definition
$I(x,\mst^r(V))=|B|$.  

Refer to \figref{packing} for what follows.  All the centers of balls in
$B$ are contained in a ball of radius $2r$ centered at $x$.  Therefore,
a simple packing argument implies that there exists a ball, $b$, of
radius $r/2$ that contains at least $|B|/5^d$ centers of balls in $B$.
($5^d$ is the volume of a ball of radius $5r/2$ divided by the volume
of a ball of radius $r/2$.)  The center of each of these ball is the
endpoint of an edge of length at most $2r$.  The other endpoints of these
edges are all contained in a ball of radius $5r/2$ centered around $b$.
The same packing argument shows that we can find a ball of radius $r/2$
that contains at least $|B|/(5\cdot 6)^d$ of these other endpoints.

\begin{figure}
  \begin{center}
    \includegraphics{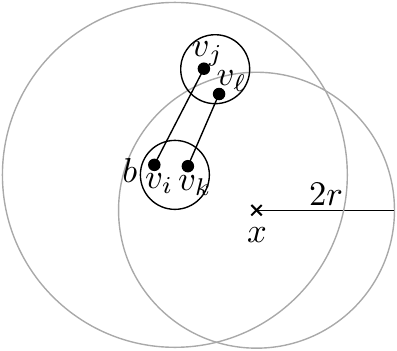}
  \end{center}
  \caption{The proof of \lemref{log}.}
  \figlabel{packing}
\end{figure}

We claim that this implies that $|B|/30^d< 2$ (so $|B|< 2\cdot
30^d$). Otherwise, $\mst(V)$ contains two edges, $x_ix_j$ and $x_kx_\ell$,
each of length greater than $r$ and such that $\|x_ix_k\| \le r$ and
$\|x_jx_\ell\|\le r$.  But this contradicts the minimality of $\mst(V)$,
since one could replace $x_ix_j$ with one of $x_ix_k$ or $x_jx_\ell$
and obtain a spanning tree of smaller total edge length.  We conclude
that $|S_i|< 2\cdot 30^d$, and this completes the proof.
\end{proof}

Note that \lemref{log} implies \thmref{log}, since it implies that we
can partition the edges of $\mst(V)$ into $\lceil\log_2 D\rceil$ classes,
based on length, and each class will contain only a constant number of edges.

We are ready to prove Parts~2 and 3 of \thmref{main}.  The sketch of the
proof is as follows:  We partition $[0,1]^d$ into equal cubes of volume
$1/nt$, for some parameter $t$ to be chosen later.  Using Chernoff's
bounds, we show that each cube contains $O((\log n)^{2/3})$ points so
that the points within each cube can be connected, using the results of
Halld\'orsson and Tokuyama, with maximum interference $O((\log n)^{1/3})$.
Next, the cubes are connected to other cubes by selecting one point in
each cube and connecting these selected points with a minimum spanning
tree.  \lemref{log} is then used to show that this minimum spanning
tree has maximum interference $O((\log n)^{1/3})$.  Without further ado,
we present:

\begin{proof}[Proof of \thmref{main}, Parts 2 and 3]
Partition $[0,1]^2$ into square \emph{cells} of area $1/nt$ for some
value $t$ to be specified later.  Let $N_i$ denote the number of points
that are contained in the $i$th cell.  Then $N_i$ is binomial
with mean $\mu=1/t$.  Recall Chernoff's Bounds \cite{c52} on the tails
of binomial random variables:
\[
  \Pr\{N_i \ge (1+\delta)\mu\} 
    \le \left(\frac{e^\delta}{(1+\delta)^{1+\delta}}\right)^\mu \enspace .
\]
In our setting, we have, 
\begin{align*}
  \Pr\{N_i \ge k\} 
    & = \Pr\{N_i \ge kt\mu\} \\
    & \le \left(\frac{e^{kt}}{(kt)^{kt}}\right)^{1/t} \\
    & = \frac{e^{k}}{(kt)^{k}} \\
    & \le \frac{1}{t^{k}} & \text{for $k\ge e$} \\
    & \le \frac{1}{n^{c+2}} \enspace , 
\end{align*}
for $t=2^{(\log n)^{1/3}}$ and $k=(c+2)(\log n)^{2/3}$.

Note that the number of cells is no more than $nt\le
n^2$, for sufficiently large $n$.  Therefore, by the union bound, the
probability that there exists any cell containing more than $k$ points
is at most $n^{-c}$.

Within each non-empty cell, we apply \thmref{sqrtn2d} to
connect the vertices in the $i$th cell into a connected graph $G_i$
with $I(G_i)=O(\sqrt{N_i})$.\footnote{This is where the discrepancy between Parts~2 and 3 of the theorem occurs.  For $d\ge 3$, \thmref{sqrtn2d} only guarantees $I(G_i)=O(\sqrt{N_i\log N_i})$.}  In fact, a somewhat stronger result holds,
namely that $\max\{I(x,G_i) : x\in\R^2\}=O(\sqrt{N_i})$.  Notice that
each edge in $G_i$ has length at most $\sqrt{2/nt}$.  Stated another
way, in $\bigcup_i G_i$, any point, $x$, receives interference only
from cells within distance $\sqrt{2/nt}$ of the cell containing $x$.
There are only 25 such cells, so
\[
  \max\left\{I\left(x,\bigcup_iG_i\right) : x\in\R^2\right\}=O(\sqrt{k}) 
    = O((\log n)^{1/3})
\]
with high probability.

Thus far, the points within each cell are connected to each other and
the maximum interference, over all points in $\R^2$, is $O(\sqrt{k})$.
To connect the cells to each other, we select one point from each
non-empty cell and connect these using a minimum spanning tree, $T$.
What remains is to show that the additional interference caused by the
addition of the edges in $T$ does not exceed $O((\log n)^{1/3})$.

Suppose that $I(x,T)=r$, for some point $x\in\R^2$.  There are at most
9 vertices in $T$ whose distance to $x$ is less than $1/\sqrt{nt}$.
Therefore, by \lemref{log}, $T$ must contain an edge of length at least
$c2^r/\sqrt{nt}$, for some constant $c>1$.  

A well-known property of minimum spanning trees is that, for any edge
$x_ix_j$ in $T$, the open ball with diameter $x_ix_j$ does not contain
any vertices of $T$.  In our setting, this means that there is an open
ball, $B$, of radius $c2^r/2\sqrt{nt}$ such that every cell contained in
$B$ contains no point of $V$.  Inside of $B$ is another empty ball $B'$
of radius $c2^r/(2\sqrt{nt})-\sqrt{2/nt}$ whose center is also the center
of some cell.

At least one quarter of the area of $B'$ is contained in $[0,1]^2$,
so the number of cells completely contained in $B'$ is at least $\pi c^22^{2r}/16 -
O(2^{r}/\sqrt{nt})$.  By decreasing $c$ slightly, and only considering
$r$ larger than a sufficiently large constant, $r_0$, we can simplify
this number of cells to $\pi c^{2r}/16$.

For a fixed ball $B'$, the probability that the $c\pi  2^{2r}/16$ cells
defined by $B'$ are empty of points in $V$ is at most
\begin{align*}
 p 
  & \le (1-c\pi 2^{2r}/{16nt})^{n} \\
  & \le \exp(-c\pi 2^{2r}/16t) \\
 & \le 1/n^{2+c'} \enspace ,
\end{align*}
for $r\ge\log(16/c\pi)+\log t + \log(2+c')+\log\ln n$.  By the union bound, the
probability that there exists any such $B'$ is at most
$pnt\le 1/n^{c'}$.  Since we can choose $r\in O(\log t+\log\log n) = O((\log
n)^{1/3})$, this completes the proof.
\end{proof}

The proof of Part~1 of \thmref{main} is just a matter of reusing the ideas
from the previous proof of Parts~2 and 3.

\begin{proof}[Proof of \thmref{main}, Part 1]
Let $x$ be any point in $\R^2$.  We partition the balls in $B(\mst(V))$
that contain $x$ into three sets:
\begin{enumerate}
  \item the set $B_0$ of balls having area at most $1/nt$;
  \item the set $B_1$ of balls having area in the range $[1/nt,(c\log
  n)/n]$; and
  \item the set $B_2$ of balls having area greater than $(c\log n)/n$.
\end{enumerate}
In this proof, the parameter $t=2^{(\log n)^{1/2}}$.

The set $B_0$ consists of points contained in a ball of area $1/nt$
centered at $x$.  Exactly the same argument used in the first part of
the previous proof shows that, with high probability, every such ball contains
$O((\log n)^{1/2})$ points, so
\[
     |B_0| \in O((\log n)^{1/2}) \enspace .
\]

The set $B_1$ consists of balls whose radii are in the range $[\sqrt{1/\pi
nt},\sqrt{(c\log n)/\pi n}]$.  \lemref{log} shows that the number of
these balls is
\begin{align*}
    |B_1| & \in O\left(\log\left(\frac{\sqrt{(c\log n)/\pi n}}{\sqrt{1/\pi nt}}\right)\right) \\
    & = O(\log\log n + \log t) \\
    & = O((\log n)^{1/2}) \enspace .
\end{align*}

Finally, any edge in the set $B_2$ implies the existence of an empty ball,
with center in $[0,1]^2$, having area $c\log n/n$.  The second part of the
previous proof shows that the probability that such a ball exists is
$O(n^{-c})$.  Therefore, with high probability,
\[
   |B_2| = 0 \enspace . \qedhere
\]
\end{proof}

\section{Proof of The Lower Bounds (\thmref{lower-bound})}

In this section, we prove the lower bounds in \thmref{lower-bound}.
We define a \emph{Zeno configuration} as follows (see \figref{zeno}):
A Zeno configuration of size $k$, centered at a point, $x$, is defined
by a set of $k+1$ balls.  The construction starts with disjoint balls
$D_0,\ldots,D_{k-1}$, each having radius $u$.  The ball $D_0$ is centered
at $x$.  The center of $D_i$, $i\in\{1,\ldots,k-1\}$ is at $x+(u3^i, 0)$.
A final large ball, $D$, of radius $r=u3^k$ is centered at $x$ and
contains all other balls.  A Zeno configuration occurs at location $x$
in a point set $V$ when $D$ contains exactly $k$ points of $V$ and these
occur with exactly one point in each ball $D_i$.

\begin{figure}
  \begin{center}
    \includegraphics[width=\textwidth]{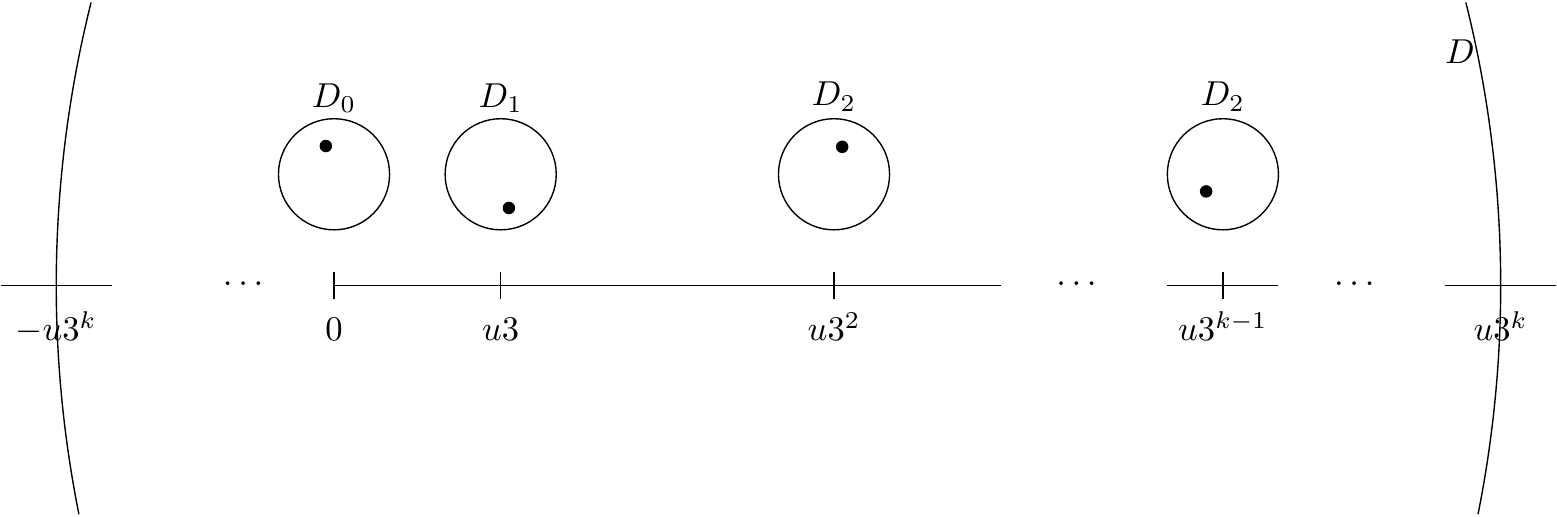}
  \end{center}
  \caption{A Zeno configuration of size $k$.}
  \figlabel{zeno}
\end{figure}

The following lemma shows that a Zeno configuration in $V$ causes high
interference in $\mst(V)$.

\begin{lem}\lemlabel{zeno-mst}
If $V$ contains a Zeno configuration of size $k$, $I(\mst(V))\ge k-1$.
\end{lem}

\begin{proof}
Let $x_i$, $i\in\{0,\ldots,k-1\}$, denote the point of $V$ contained in
$D_i$.  Note that, for $i\in\{1,\ldots,k-1\}$ the closest point to $x_i$
in $V$ is $x_{i-1}$.  Since $\mst(V)$ contains the nearest-neighbour
graph, this implies that $\mst(V)$ contains the edges $x_ix_{i+1}$ for
all $i\in\{0,\ldots,k-2\}$.  See \figref{zeno-mst} for what follows.
We claim that, for all $i\in\{0,\ldots,k-2\}$, the ball $B_i$ centered
at $x_i$ that contains $x_{i+1}$ also contains $x_0$.  This is clearly
true for $i=0$ and $i=1$.  Next, note that
\[
  \|x_ix_0\| \le u(3^i+2) \enspace .
\]
On the other hand, for $i\ge 2$,
\[
  \|x_ix_{i+1}\| \ge u(3^{i+1}-3^i-2) = 2u3^i-2u \ge u(3^i + 7) > 
\|x_ix_0\| \enspace .
\]
Therefore, $I(x_0,\mst(V)) \ge k-1$.
\end{proof}

\begin{figure}
  \begin{center}
    \includegraphics{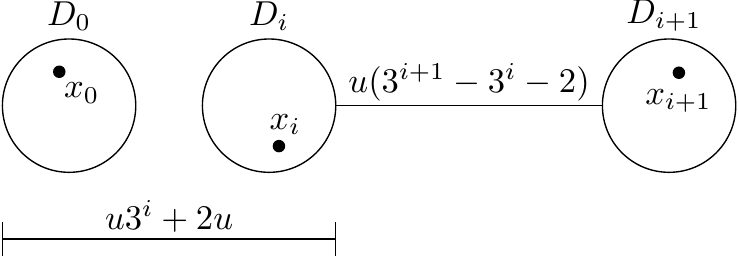}
  \end{center}
  \caption{The ball centered at $x_i$ that contains $x_{i+1}$ also contains $x_0$.}
  \figlabel{zeno-mst}
\end{figure}

The next lemma shows that a Zeno configuration causes high interference on
any connected graph on vertex set $V$.

\begin{lem}
If $V$ contains a Zeno configuration of size $k$, then $I(V)\ge\sqrt{k-1}$.
\end{lem}

\begin{proof}
Let $G$ be any connected graph on $V$.  Using the same notation as
in the proof of \lemref{zeno-mst}, call a vertex, $x_i$, a \emph{big
one} if $x_i$ is adjacent to any vertex $x_j$, with $j>i$, or $x_i$ is
adjacent to any vertex $x$ not in $D$.  The proof of \lemref{zeno-mst}
shows that every big one contributes to the interference at $x_0$.
Therefore, if the Zeno configuration contains $\sqrt{k-1}$ or more big
ones, then $I(x_0,G)\ge\sqrt{k-1}$ and there is nothing left to prove.
Otherwise, note that each of $x_0,\ldots,x_{k-2}$ is either a big one
or adjacent to a big one. Therefore, there must be a big one, $x_i$,
with degree at least $\sqrt{k-1}-1$, so $I(x_i,G)\ge\sqrt{k-1}$.
\end{proof}

To prove \thmref{lower-bound}, all that remains is to show a Zeno
configuration of size $\Omega((\log n)^{1/2})$ occurs in $V$ with high
probability.

\begin{proof}[Proof of \thmref{lower-bound}]
Choose the parameter $u$ in the Zeno configuration so that $\pi r^2=1/n$,
i.e., $u=1/(\sqrt{\pi n}3^k)$.  Then the area of the small balls is
$\pi u^2=1/(n3^{2k})$.
We analyze the probability that a Zeno configuration of length $k=c(\log
n)^{1/2}$ centered at $x_i$ occurs in a set, $V$, of $n$ i.u.d. points
$\{x_1,\ldots,x_n\}$ in
$[0,1]^2$.  Let $\mathcal{Z}_i$ denote the event ``$V$
contains a Zeno configuration centered at $x_i$.''  Then we have
\begin{align*}
 \Pr\{\mathcal{Z}_i\mid x_i\in[r,1-r]^2\} 
  & = \left(\frac{(n-1)!}{(n-k)!}\right) 
      \left(\frac{1}{n3^{2k}}\right)^{k-1}   
      \left(1-\frac{1}{n}\right)^{n-k}   \\ 
  & \ge
      \left((n-k)^{k-1}\right) 
      \left(\frac{1}{n3^{2k}}\right)^{k-1}  
      \left(1-\frac{1}{n}\right)^{n-k}   \\ 
  & \ge 
      (1-k/n)^{k-1} %
      \left(\frac{1}{3^{2k(k-1)}}\right)
      \left(1-\frac{1}{n}\right)^{n-k} \\
  & \ge
      \left(1-o(1)\right) 
      \left(\frac{1}{3^{2k(k-1)}}\right)
      \left(1/e-o(1)\right)   \\ 
  & \ge (1/e-o(1)) \left(\frac{1}{3^{2k(k-1)}}\right) \\
  & = \Omega(1/n^{\alpha})
\end{align*}
for $k=((\alpha/2)(\log_3 n))^{1/2}$, where $\alpha$ is a free parameter
in the range $[0,1]$.  Since $\Pr\{x_i\in[r,1-r]^2\} > 1-4r$, we now
uncondition
\[
 \Pr\{\mathcal{Z}_i\} \ge  
    (1-4r)\cdot\Pr\{\mathcal{Z}_i\mid x_i\in[r,1-r]^2\}
    = \Omega(1/n^{\alpha})
\]
Let $Y_i$ be the indicator variable defined as
\[
   Y_i = \begin{cases} 1 & \text{if $\mathcal{Z}_i$} \\
                       0 & \text{otherwise} 
         \end{cases}
\]
and let $N=\sum_{i=1}^n Y_i$ count the number of Zeno configurations.  We
have just shown that 
\[
   \E[N] = n\E[Y_i] = n\Pr\{\mathcal{Z}_i\}=\Omega(n^{1-\alpha}) \enspace .
\]
Unfortunately, this is not quite enough to prove that $N>0$ with high
probability.  Instead, we finish the proof using the second moment method
(c.f., Alon and Spencer \cite[Chapter~4]{as08}).  For this, we need only 
show that, for any $\{i,j\}\subset\{1,\ldots,n\}$,
\[
   \limsup_{n\rightarrow\infty}\frac{\E[Y_iY_j]}{\E[Y_i]\E[Y_j]} = 1
   \enspace .  \footnote{In particular, this shows that the value
  $\Delta$ in Ref.~\cite[Corollary~4.3.4]{as08} satisfies the condition $\Delta\in o(\E[N]^2)$.}
\]
To do this, we repeat the above argument, but for a pair of Zeno
configurations, one at $x_i$ and one at $x_j$.  Let $A$ denote the event
``$\|x_ix_j\| < 2r$ or $\{x_i,x_j\}\not\subset[r,1-r]^2$''.  Let $A^c$
denote the complement of $A$.  Conditioning on $A^c$ we obtain
\begin{align*}
\frac{\E[Y_iY_j]}{\E[Y_i]\E[Y_j]} 
& \le \frac{\E[Y_iY_j]}{1/(e3^{2k(k-1)})^2} \\
& = (e3^{2k(k-1)})^2(\Pr\{A^c\}\E[Y_iY_j|A^c] + \Pr\{A\}\E[Y_iY_j|A]) \\
& \le (e3^{2k(k-1)})^2\E[Y_iY_j|A^c]  + (4r+\pi r^2) \\
& \le (e3^{2k(k-1)})^2 
    \left(\frac{(n-2)!}{(n-2k)!}\right)
    \left(\frac{1}{n3^{2k}}\right)^{2k-2}
    \left(1-\frac{2}{n}\right)^{n-2k} + (4r+\pi r^2) \\
& \le (e3^{2k(k-1)})^2
    \left(n^{2k-2}\right)
    \left(\frac{1}{n3^{2k}}\right)^{2k-2}
    \left(1-\frac{2}{n}\right)^{n-2k} + (4r+\pi r^2) \\
& \le (e3^{2k(k-1)})^2
    \left(\frac{1}{3^{2k}}\right)^{2k-2}
    \left(1-\frac{2}{n}\right)^{n-2k} + (4r+\pi r^2) \\
& \le (e3^{2k(k-1)})^2
    \left(\frac{1}{3^{2k}}\right)^{2k-2}
    \left(1/e^2-o(1)\right) + (4r+\pi r^2) \\
& = e^23^{4k(k-1)}
    \left(\frac{1}{3^{4k(k-1)}}\right)
    \left(1/e^2-o(1)\right) + (4r+\pi r^2) \\
& = 1-o(1) + O(1/\sqrt{n}) \rightarrow 1 \enspace ,
\end{align*}
as $n\rightarrow\infty$.  This completes the proof.
\end{proof}

\section{Discussion}

\paragraph{Summary.}
This paper gives new bounds on the maximum interference for graphs defined
by points randomly distributed $[0,1]^d$. Minimum spanning trees have
interference $\Theta((\log n)^{1/2})$, but better graphs exist; a strategy
based on bucketing yields a graph with interference $O((\log n)^{1/3})$.
No graph on such a point set has interference $o((\log n)^{1/4})$.

\paragraph{Open Problem.}
An obvious open problem is that of closing the gap between the upper bound
of $O((\log n)^{1/3})$ and the lower bound of $\Omega((\log n)^{1/4}$.
One strategy to achieve this would be to prove the following conjecture,
which has nothing to do with probability theory:
\begin{conj}
  For any $V\subset\R^d$, $I(V) = O(\sqrt{I(\mst(V))})$.
\end{conj}
A weaker version of this conjecture is due to Halld\'orsson and Tokuyama
\cite{ht08}, who conjecture that $I(V)=O(\sqrt{\log D})$ where $D$ is the
ratio of the lengths of the longest and the shortest edges of $\mst(V)$.

\paragraph{Unit Disk Graphs.}
Several of the references consider interference in the \emph{unit
disk graph model}, in which the graph $G$ is constrained to use
edges of maximum length $r(n)$.  It is straightforward to verify
that all of the proofs in this paper continue to hold in this model,
when $r(n)\in\Omega(\sqrt{(\log n)/n})$.  This is not an unreasonable
condition; for i.u.d.\ points in $[0,1]^d$, it is known that
$r(n)\in\Omega(\sqrt{(\log n)/n})$ is a necessary condition to be able to
form a connected graph $G$ \cite{p97}.

\paragraph{Locally Computable Graphs.}
Khabbazian, Durocher, and Haghnegahdar \cite{kdh11} give a local
algorithm, called \textsc{LocalRadiusReduction}, that is run at the nodes
of a communication graph, $G=(V,E)$, and that reduces the number of edges
of $G$.  The resulting graph $G'$ comes from a class of graphs that
they denote as $\mathcal{T}(V)$.  The class $\mathcal{T}(V)$ includes
the minimum spanning tree of $V$ and the graphs in this class share
many of the same properties as the minimum spanning tree.  In particular,
the following result can be obtained by using the proof of \thmref{main} Part 1 and properties of the family $\mathcal{T}(V)$ \cite[Theorem~3]{kdh11}.

\setcounter{thm}{2}
\begin{thm}\thmlabel{tp}
  Let $V$ be a set of $n$ independently and uniformly distributed
  points in $[0,1]^d$ and let $G$ be any graph in $\mathcal{T}(V)$.
  With high probability, $I(G)=O((\log n)^{1/2}+\log (\ell\sqrt{n}))$,
  where $\ell$ is the length of the longest edge in $G$.
\end{thm}

In particular, \thmref{tp} implies that running the
\textsc{LocalRadiusReduction} algorithm at the nodes of a unit disk graph
with unit $r(n)\in O(2^{\sqrt{\log n}}/\sqrt{n})$ yields a connected
graph with maximum interference $O((\log n)^{1/2})$.

\section*{Acknowledgement}

The research in this paper was started at the workshop on \emph{Models
of Sparse Graphs and Network Algorithms (12w5004)}, hosted at the Banff
International Research Station (BIRS), February 5--10, 2012.  The authors
are grateful to the other workshop organizers, Nicolas~Broutin and
G\'abor~Lugosi, the other participants, and the staff at BIRS, for
providing a stimulating research environment.

\bibliographystyle{plain}
\bibliography{iud}

\end{document}